\documentclass{article}
\usepackage{amsfonts}
\usepackage{amsmath}
\usepackage{amssymb}
\usepackage{graphicx}
\providecommand{\U}[1]{\protect\rule{.1in}{.1in}}
\newtheorem{theorem}{Theorem}

\newtheorem{proposition}[theorem]{Proposition}
\newtheorem{remark}[theorem]{Remark}

\newenvironment{proof}[1][Proof]{\noindent\textbf{#1.}}{\ \rule{0.5em}{0.5em}}
\begin{document}

\title{ The R\'{e}nyi entropy of L\'{e}vy distribution}
\author{Giorgio SONNINO\\Department of Theoretical Physics and Mathematics, Universit\'{e}\\Libre de Bruxelles (ULB), Campus Plain CP 231, Boulevard de \\Triomphe, 1050 Brussels, Belgium.\\ Royal Military School (RMS),\\Av. de la Renaissance 30, 1000 Brussels, Belgium\\ Email: gsonnino@ulb.ac.be
\and Gy\"{o}rgy STEINBRECHER\\Association EURATOM-MEdC, University of Craiova, \\ Physics Department, A. I. Cuza 13, 200585 Craiova, Romania. \\ Email: gyorgy.steinbrecher@gmail.com 
\and Alberto SONNINO\\ Ecole Polytechnique de Louvain (EPL) \\ Universit{\'e} Catholique de Louvain (UCL), Rue Archim$\grave{\rm e}$de\\1 bte L6.11.01, 1348 Louvain-la-Neuve - Belgium. \\ Email: alberto.sonnino@gmail.com}
\maketitle

\begin{abstract}
The equivalence between non-extensive C. Tsallis entropy and the

extensive entropy introduced by Alfr\'{e}d R\'{e}nyi \ is discussed. The
R\'{e}nyi entropy is studied from the perspective of the geometry of the
Lebesgue and generalised, exotic Lebesgue spaces. A duality principle is
established. The R\'{e}nyi entropy for the L\'{e}vy distribution, in the
domain when the nunerical methods fails, is approximated by asymptotic
expansion for the large values of the R\'{e}nyi parameter.

\end{abstract}

\section{Introduction}

The independent discovery of the equivalent generalizations of the
Shannon-Boltzmann entropy by A. R\'{e}nyi \cite{Renyi1}, by pure axiomatic
reasoning respectively by C. Tsallis, from physical considerations has a great
influence both on statistical physics as well as on mathematical methods in
statistics, engineering and informatics. In the case of systems with small
degree of freedom the entropy of C. Tsallis and A. R\'{e}nyi are completely
equivalent, in the sense that the Tsallis entropy can be computed easily when
the R\'{e}nyi entropy is known and conversely. The axioms that uniquely
determines he R\'{e}nyi entropy can be exactly reformulated for the Tsallis
entropy, despite the R\'{e}nyi formulation is closely related to the formalism
of topological metric vector spaces. \ The MaxEnt principle formulated for the
R\'{e}nyi entropy can be reformulated as a problem of geometry on convex sets
in Banach or complete metric linear spaces. The result of computation of
probability density functions (PDF) by MaxEnt principle with or without
restrictions give the same PDF both for R\'{e}nyi and Tsallis. Nevertheless
the entropy introduced by A. R\'{e}nyi, has many elegant physical and
mathematical properties. First of all, the \ R\'{e}nyi entropy is are
additive, so the study of the thermodynamic limit of non interacting or weakly
interacting subsystems can be studied by suitable perturbation methods . We
can introduce also the intensive quantity, the R\'{e}nyi entropy per particle,
or per volume. Despite this property ( the possibility to express the entropy
of composed system by entropy of subsystems) can be transferred to the Tsallis
entropy, the formula is quite complicated . When the number of non-interacting
subsystems is of the order of magnitude of the Avogadro number, the numerical
value of the Tsallis entropy contains a very large exponent. \ In the case of
the continuos distributions, the R\'{e}nyi entropy has several new advantages
over the Tsallis entropy. First, at the rescaling of the variables, by the
change of units measurements, the R\'{e}nyi entropy changes by an additive
constant, exactly in the same way as the classical Shannon Boltzmann Gibbs
entropy. \ Moreover, when using R\'{e}nyi entropy , we can use instead of the
probability density function the particle density function, or in the case of
identical particles, the mass density function. In all of these cases the
variation of R\'{e}nyi entropy remains unchanged, like in the case of the
Shannon-Boltzmann-Gibbs entropy. 

The adaptation of the Tsallis entropy to these changes (scaling of variables,
replacing PDF by particle density function, study of systems composed of
subsystems whose number is of the order of magnitude of Avogadro number) is
still possible. In fact all of the reasonings performed in the framework of
the R\'{e}nyi entropy can be translated easily in the language of the Tsallis
entropy. In particular in the case of composing a large system from
independent subsystems, the addition of the R\'{e}nyi entropy is translated in
a new, apparently non linear operation, that allows to compute the Tsallis
entropy of the composed system. So it remains at the taste of the researchers
the choice. 

In conclusion, the physical arguments for the use of Tsallis entropy
(\cite{Tsallis1}) are valid also for the use of the R\'{e}nyi entropy, as well
as the Markov chain argument (\cite{Renyi1}) for the use of the R\'{e}nyi
entropy applies to Tsallis entropy.

Finally, we can relate the R\'{e}nyi entropy  to well studied concepts in
functional analysis: Banach space norms or pseudo-norms in complete metric
vector spaces. The interpretation of the R\'{e}nyi entropy in the language of
convex analysis and geometry in $L_{p}$ spaces , allows to use mathematical
results, that could give a simple proof to useful properties of the Tsallis
entropy. 

The structure of this work is the following. In Section \ref{markSect2}, for
the sake of self consistency, we expose the definition and main properties of
the R\'{e}nyi entropy. In Section \ref{MarkSect3} we discuss the properties of
the R\'{e}nyi entropy from the viewpoint of geometry of Banach spaces
respectively from the perspective of the geometry of\ exotic, non local convex
Lebesgue spaces with exponent $0<p<1$. Here, in the subsection
\ref{MarkSubsectDuality} we establish a new result, the duality property of
the R\'{e}nyi entropy, that we consider that will be useful in numerical
methods for performing optimized approximation of the solutions of the
Fokker-Planck equations. In the Section \ref{MArkSect4Levy} we compute the
R\'{e}nyi entropy for the Levy distribution in a domain when the numerical
methods fails.

\section{The equivalence of entropy of C. Tsallis and A.
R\'{e}nyi\label{markSect2}}

\subsection{Discrete distribution}

In the simplest possible case of discrete probability field with $N$ states
and with associated probabilities $\{p_{1},...,p_{N}\}$ , the Tsallis
respectively R\'{e}nyi entropies, with index $q>1$, are defined as follows%
\begin{align}
S_{T,q}(p_{i})  &  =\frac{1}{1-q}\left(  1-%
%TCIMACRO{\dsum \limits_{i=1}^{N}}%
%BeginExpansion
{\displaystyle\sum\limits_{i=1}^{N}}
%EndExpansion
p_{i}^{q}\right) \label{1}\\
S_{R,q}(p_{i})  &  =\frac{1}{1-q}\log\left(
%TCIMACRO{\dsum \limits_{i=1}^{N}}%
%BeginExpansion
{\displaystyle\sum\limits_{i=1}^{N}}
%EndExpansion
p_{i}^{q}\right)  \label{2}%
\end{align}
It is easy to see that in the limit case $q\rightarrow1$ both $S_{T,q}(p_{i})$
and $S_{R,q}(p_{i})$ approach the Shannon entropy.%
\[
S_{Shannon}(p_{i})=-%
%TCIMACRO{\dsum \limits_{i=1}^{N}}%
%BeginExpansion
{\displaystyle\sum\limits_{i=1}^{N}}
%EndExpansion
p_{i}\log p_{i}%
\]
Their equivalence, when discussing the problem of finding the probability
distributions that obey maximal entropy principle with some restrictions,
results in a straightforward manner from Eqs(\ref{1}, \ref{2})
\begin{equation}
S_{T,q}(p_{i})=\frac{1}{\ q-1}\left[  \exp\left(  \left(  q-1\right)
S_{R,q}(p_{i})\right)  -1\right]  \label{3}%
\end{equation}

\begin{remark}
\label{RemTsalisEqRennyi}Observe that Eq.(\ref{3}) will be valid also in the
case of continuos, multi variate distributions
\end{remark}

In the subsequent part, we will discuss the R\'{e}nyi entropies also for the
case $0<q<1$. It is clear from Eq.(\ref{3}) that in the case of the small
system both entropies can be computed easily and give the same amount in
information. In particular the dependence of $S_{T,q}(p_{i})$ over
$S_{R,q}(p_{i})$ is monotonous, for all $0<q$, so the maximal entropy
principles, eventually with restrictions, in both case give rise to the same
probabilities. Moreover, in the R\'{e}nyi formulation, we can consider the
limit $q\rightarrow\infty$ too. From Eq.(\ref{2}) in the limit $q\rightarrow
\infty$ we still obtain a meaningful result:
\[
S_{R,q}(p_{i})=-\frac{q}{q-1}\log\left(
%TCIMACRO{\dsum \limits_{i=1}^{N}}%
%BeginExpansion
{\displaystyle\sum\limits_{i=1}^{N}}
%EndExpansion
p_{i}^{q}\right)  ^{1/q}\underset{q\rightarrow\infty}{\longrightarrow}%
-\log\left[  \underset{i}{\max}\left(  p_{1},...,p_{n}\right)  \right]  >0
\]

Observe that in this limit the maximal entropy principle give rise to the
min-max problem with restrictions In the case of the system composed of large
number of identical and approximately independent subsystems, the use of
$S_{R,q}$ is more suitable, because it is extensive and we can define the
intensive "specific R\'{e}nyi entropy". Indeed consider a complex system with
$M\times N$ states specified by \ the probabilities $\{p_{i,j}|1\leq i\leq
N;1\leq j\leq M\}$ . In this case the R\'{e}nyi entropy is%
\begin{equation}
S_{R,q}(p_{i,j})=\frac{1}{1-q}\log\left(
%TCIMACRO{\dsum \limits_{i=1}^{N}}%
%BeginExpansion
{\displaystyle\sum\limits_{i=1}^{N}}
%EndExpansion%
%TCIMACRO{\dsum \limits_{j=1}^{M}}%
%BeginExpansion
{\displaystyle\sum\limits_{j=1}^{M}}
%EndExpansion
p_{i,j}^{q}\right)  \label{4}%
\end{equation}

In the case when the large system is composed of independent subsystems with
$N$ respectively $M$ states with probabilities $\{p_{1}^{\prime}%
,...,p_{N}^{\prime}\}$ respectively $\{p_{1}^{^{\prime\prime}},...,p_{M}%
^{^{\prime\prime}}\}$ we have
\[
p_{i,j}=p_{i}^{\prime}p_{j}^{^{\prime\prime}}%
\]

and from Eq.(4) results
\begin{equation}
S_{R,q}(p_{i,j})=S_{R,q}(p_{i}^{\prime})+S_{R,q}(p_{i}^{\prime\prime})
\label{5}%
\end{equation}

In the continuation, due to Eqs.(\ref{3} , \ref{5} ) we discuss the R\'{e}nyi
entropy. \ The Eq.(\ref{5}) can be easily generalized, so in the case of a
large system composed of a large number of $\mathcal{N}$ identical and
independent (non-interacting) subsystems , each of them with $N$ states and
with associated probabilities $\{p_{1},...,p_{N}\}$ , the entropy
$S_{\mathcal{N},q}$ of the resulting large system is simply%
\begin{equation}
S_{\mathcal{N},q}=\mathcal{N}\ S_{R,q}(p_{i}) \label{6}%
\end{equation}
which is advantageous when $\mathcal{N}\ \ $\ is of the order of magnitude of
Avogadro number. For physical systems with short\ range or weak interactions,
it is meaningful to attempt to compute $S_{\mathcal{N},q}/\mathcal{N}$ in the
thermodynamic limit.

Like Shannon entropy, the R\'{e}nyi entropy is determined uniquely by a system
of physically meaningful axioms (\cite{Renyi1}).

\section{General distributions, R\'{e}nyi entropy and Lebesgue space norm\label{MarkSect3}}

Suppose we have a measure space $(\Omega,\mathcal{\ }m)$ with positive not
necessary finite, measure $m$ over the space $\Omega$ . For sake of simplicity
the $\sigma$ algebra will be omitted. Suppose we have the probability density
function $\rho(\mathbf{x})$ on $\Omega$ , where\ $\mathbf{x}\in\Omega$, such
that
\begin{equation}%
%TCIMACRO{\dint \limits_{\Omega}}%
%BeginExpansion
{\displaystyle\int\limits_{\Omega}}
%EndExpansion
\rho(\mathbf{x})dm(\mathbf{x})=1 \label{6.1}%
\end{equation}

\begin{remark}
\label{RemNonZeroProbOnOmega}Without loss of generality we can consider that
$\rho(\mathbf{x})>0$ excepting a null set $A$ that is statistically
negligible, i.e. $m(A)=0\,$
\end{remark}

In this case the Tsallis and R\'{e}nyi entropies, for $q>0$ and $q\neq1$ are
defined as follows, if the integrals exists%

\begin{align}
S_{T,q}(\rho)  &  =\frac{1}{1-q}\left(  1-%
%TCIMACRO{\dint \limits_{\Omega}}%
%BeginExpansion
{\displaystyle\int\limits_{\Omega}}
%EndExpansion
\left[  \rho(\mathbf{x})\right]  ^{q}dm(\mathbf{x})\ \right) \label{6.2}\\
S_{R,q}(\rho)  &  =\frac{1}{1-q}\log\left[
%TCIMACRO{\dint \limits_{\Omega}}%
%BeginExpansion
{\displaystyle\int\limits_{\Omega}}
%EndExpansion
\left[  \rho(\mathbf{x})\right]  ^{q}dm(\mathbf{x})\right]  \label{6.3}%
\end{align}

Observe that the relation Eq.( \ref{3} ) between the entropies remains valid
.  From the Eq.(\ref{6.3}) results:%
\begin{align}
S_{R,q}(\rho) &  =\frac{q}{1-q}\log\left[  \left\Vert \rho\right\Vert
_{L_{q}(\Omega,dm)}\right]  ;q>1\label{6.31}\\
S_{R,q}(\rho) &  =\frac{1}{1-q}\log N_{L_{q}(\Omega,dm)}(\rho
);0<q<1\label{6.311}%
\end{align}

where the notations from (\cite{ReedSimon}, \cite{Rudin}) \ \ are used.

The unusual $L_{q}(\Omega,dm)$ spaces with $0<q<1$ are important for the study
of the distributions with heavy tail (\cite{LuschgiPages}, \cite{SGW}) .

It is convenient to define the probability measure
\begin{equation}
dP(x)=\rho(\mathbf{x})dm(\mathbf{x}) \label{6.32}%
\end{equation}
that has the property%
\begin{equation}%
%TCIMACRO{\dint \limits_{\Omega}}%
%BeginExpansion
{\displaystyle\int\limits_{\Omega}}
%EndExpansion
dP(\mathbf{x})=1 \label{6.4}%
\end{equation}

The R\'{e}nyi entropy functional, according to Eq.(\ref{6.3}), associated to
this distribution is (see Remark(\ref{RemNonZeroProbOnOmega}))
\begin{equation}
S_{q}\{\rho\}=-\log\left[
%TCIMACRO{\dint \limits_{\Omega}}%
%BeginExpansion
{\displaystyle\int\limits_{\Omega}}
%EndExpansion
\left[  \rho(x)\right]  ^{q-1}dP(x)\right]  ^{\frac{1}{q-1}} \label{7}%
\end{equation}

The following proposition results from Eq.(\ref{7}) and are well known.

\begin{proposition}
\begin{remark}
\label{PropMonotony}In the case $\ $where $\int\limits_{\Omega}dP(x)=1$, i.e.
$P$ is a probability measure, then for $0<a_{1}<a_{2}$ from the H\"{o}lder
inequality (\cite{ReedSimon}, \cite{Rudin}) results
\begin{equation}
\left[
%TCIMACRO{\dint \limits_{\Omega}}%
%BeginExpansion
{\displaystyle\int\limits_{\Omega}}
%EndExpansion
|f(\mathbf{x})|^{a_{1}}dP(\mathbf{x})\right]  ^{1/a_{1}}\leq\left[
%TCIMACRO{\dint \limits_{\Omega}}%
%BeginExpansion
{\displaystyle\int\limits_{\Omega}}
%EndExpansion
|f(\mathbf{x})|^{a_{2}}dP(\mathbf{x})\right]  ^{1/a_{2}} \label{11}%
\end{equation}

\end{remark}
\end{proposition}

\begin{proposition}
\label{PropEntropyMonotony}The R\'{e}nyi and Tsallis entropies for
$0<q\ $decreases with respect to $q$ .
\end{proposition}

\begin{proof}
In the case $1<q_{1}<q_{2}$ we use Eq.(\ref{11}) with $a_{1}=q_{1}%
-1,a_{2}=q_{2}-1$ and $f(\mathbf{x})=\rho(\mathbf{x})$ \ combined with
Eq.(\ref{7}). In the case $0<q_{1}<q_{2}<1$ we have (See Remark
(\ref{RemNonZeroProbOnOmega}))%
\begin{equation}
S_{q}\{\rho\}=\log\left[
%TCIMACRO{\dint \limits_{\Omega}}%
%BeginExpansion
{\displaystyle\int\limits_{\Omega}}
%EndExpansion
\left[  \frac{1}{\rho(x)}\right]  ^{1-q}dP(x)\right]  ^{\frac{1}{1-q}}
\label{11a}%
\end{equation}
By using Eq.(\ref{11}) with $a_{1}=1-q_{2},a_{2}=1-q_{1}$ and $f(\mathbf{x}%
)=1/\rho(\mathbf{x})$ \ combined with Eq.(\ref{11a}) completes the proof. The
corresponding property for the Tsallis entropy results from Eq.(\ref{3}) and
Remark(\ref{RemTsalisEqRennyi}).
\end{proof}

Now it is easy to prove the following

\begin{proposition}
\label{PropRenyiApprShannon}In the limit $q\searrow1$ the R\'{e}nyi entropy
approaches the Shannon entropy .
\end{proposition}

\begin{proof}
From Proposition (\ref{PropMonotony}) results that the limit $\underset
{q\searrow1}{\lim}\left\Vert \rho\right\Vert _{L_{q-1}(\Omega,dP)}$
\ \ exists. In the the case when $\rho(\mathbf{x})>0$ the function
$f(q)=\left\Vert \rho\right\Vert _{L_{q-1}(\Omega,dP)}$ is analytic \ near
$q=1$. By series expansion in Eq.(\ref{9}) and Eq.( \ref{6.32} )results that%
\[
\underset{q\searrow1}{\lim}S_{q}\{\rho\}=-%
%TCIMACRO{\dint \limits_{\Omega}}%
%BeginExpansion
{\displaystyle\int\limits_{\Omega}}
%EndExpansion
\rho(\mathbf{x})\log[\rho(\mathbf{x})]dm(\mathbf{x})=S_{Shannon}\{\rho\}
\]
The the extension to the case when on some domains $\rho(\mathbf{x})=0$ is
obtained by approximating the PDF $\rho(\mathbf{x})$with non zero PDF . By
extending the function $f(x)=x\log x$ by continuity such that $f(0)=0$
completes the proof.$\ $
\end{proof}

\subsection{Duality properties\label{MarkSubsectDuality}}

Let $\ p\geq1$ and $\ q\geq1$ such that $1/p+1/q=1$. We define the scalar
product of the functions $f(x)\in L_{p}(\Omega,dm)$ and $g(x)\in L_{q}%
(\Omega,dm)$ as follows%
\begin{equation}
\left\langle f,g\right\rangle =%
%TCIMACRO{\dint \limits_{\Omega}}%
%BeginExpansion
{\displaystyle\int\limits_{\Omega}}
%EndExpansion
f(\mathbf{x})g(\mathbf{x})]dm(\mathbf{x}) \label{11.01}%
\end{equation}

From duality between $L_{p}(\Omega,dm)$ and $L_{q}(\Omega,dm)$ we
have(\cite{ReedSimon}, \cite{Rudin})
\begin{equation}
\left\Vert f\right\Vert _{L_{p}(\Omega,dm)}=\underset{g\in B_{q}}{\sup
}\left\vert \left\langle f,g\right\rangle \right\vert \label{11.02}%
\end{equation}

where $B_{q}$ is the unit sphere in the space $L_{q}(\Omega,dm):$%
\begin{equation}
B_{q}=\left\{  g|g\in L_{q}(\Omega,dm);~\left\Vert g\right\Vert _{L_{q}%
(\Omega,dm)}=1\right\}  \label{11.03}%
\end{equation}

From Eqs. (\ref{11.02}, \ref{11.03}) and from the relation between norm and
entropy Eq.(\ref{6.31}), we obtain the following duality between entropies%
\begin{align}
S_{p}\{\rho\}  &  =\frac{p}{1-p}\underset{\rho^{\prime}\in B_{q}}{\sup
}\left\langle \rho,\rho^{\prime}\right\rangle \label{11.04}\\
\frac{1}{p}+\frac{1}{q}  &  =1
\end{align}

This relation can be reformulated in the term of entropies : We have
$\{\rho^{\prime}|\rho^{\prime}\in B_{q}\}=\{\rho^{\prime}|S_{q}(\rho^{\prime
})=0\}$ so we obtain from Eq.(11.04%
\begin{equation}
S_{p}\{\rho\}=\frac{p}{1-p}\underset{S_{q}(\rho^{\prime})=0}{\sup}\left\langle
\rho,\rho^{\prime}\right\rangle ;~\frac{1}{p}+\frac{1}{q}=1 \label{11.05}%
\end{equation}

\section{The R\'{e}nyi entropy of the symmetric L\`{e}vy
distributions\label{MArkSect4Levy}}

The L\`{e}vy  distribution plays \ an important role in the study of processes
with heavy tail PDF, that are encountered in the statistical physics of
anomalous transport processes. In this section we compute the R\'{e}nyi
entropy of the PDF of Levy distribution in the range of large values of the
R\'{e}nyi $q$ parameter, the domain when the numerical method does not works.

The generating function $g_{\alpha}(t)$ of the symmetric $\alpha$ stable
distribution (L\`{e}vy), with convenient normalization, having PDF
$\rho_{\alpha}(x)$ , is given by%
\begin{equation}
g_{\alpha}(t)\equiv\int\limits_{-\infty}^{\infty}\rho_{\alpha}(x)\exp\left(
itx\right)  dx=\exp\left(  -k\left\vert t\right\vert ^{\alpha}\right)
\label{L1}%
\end{equation}

Because the stability of the distribution specified by the generating function
$g_{\alpha}(t)$ is related to the variable $x$, in the sense that the sum of
independent variables distributed according to L\`{e}vy distribution is again
from the same family, it is natural to use the Lebesgue measure for the
definition of the R\'{e}nyi entropy%
\begin{equation}
S_{R,q}(\rho_{\alpha})=\int\limits_{-\infty}^{\infty}\left[  \rho_{\alpha
}(x)\right]  ^{q}dx\label{L2}%
\end{equation}

The inversion of the Fourier transform in Eq.(\ref{L1}) cannot be obtained via
elementary functions, nevertheless we can obtain easily the first terms in the
Taylor expansion near $x=0$ as follows. From Eq.(\ref{L1}) results (with the
normalization $k=1$)
\begin{equation}
\rho_{\alpha}(x)=\frac{1}{2\pi}\int\limits_{-\infty}^{\infty}\exp\left(
-itx-\left\vert t\right\vert ^{\alpha}\right)  dx\label{L3}%
\end{equation}

From Eq.(\ref{L3}) results that for $\alpha>1$ the PDF $\rho_{\alpha}(x)$ is
in fact an entire analytic function in the variable $x$ , for $\alpha=1$ is
meromorphic with poles at $x=\pm i$ and for $0<\alpha<1$ is from the
$C^{\infty}$ class. Consequently from Eq.(\ref{L3}), after taking into account
the symmetry, we obtain
\begin{align}
\rho_{\alpha}(0) &  =\frac{1}{\pi}\int\limits_{0}^{\infty}\exp\left(
-\left\vert t\right\vert ^{\alpha}\right)  dx=\frac{1}{\pi\alpha}\Gamma
(\frac{1}{\alpha})\label{L4}\\
\left[  \frac{d\rho_{\alpha}(x)}{dx}\right]  _{x=0} &  =\left[  \frac
{d^{3}\rho_{\alpha}(x)}{dx^{3}}\right]  _{x=0}=0\label{L5}\\
\left[  \frac{d^{2}\rho_{\alpha}(x)}{dx^{2}}\right]  _{x=0} &  =-\frac{1}%
{\pi\alpha}\Gamma(\frac{3}{\alpha})\label{L6}\\
\left[  \frac{d^{4}\rho_{\alpha}(x)}{dx^{4}}\right]  _{x=0} &  =-\frac{1}%
{\pi\alpha}\Gamma(\frac{5}{\alpha})\label{L7}%
\end{align}
The integral from Eq.(\ref{L2}) can be rewritten as follows%
\begin{align}
S_{R,q}(\rho_{\alpha}) &  =\int\limits_{-\infty}^{\infty}\exp\left[
f_{\alpha}(x)\right]  dx\label{L8}\\
f_{\alpha}(x) &  =\log\rho_{\alpha}(x)\label{L9}%
\end{align}

The asymptotic expansion of $S_{R,q}(\rho_{\alpha})$ for the large values of
the R\'{e}nyi parameter $q$ can be obtained according to Laplace method,
taking into account that according to Eqs.(\ref{L9}, \ref{L5})%
\[
\left[  \frac{df_{\alpha}(x)}{dx}\right]  _{x=0}=\left[  \frac{d^{3}f_{\alpha
}(x)}{dx^{3}}\right]  _{x=0}=0
\]
Consequently we obtain
\begin{equation}
S_{R,q}(\rho_{\alpha})\underset{q\rightarrow+\infty}{\asymp}\frac{\exp\left[
qf(0)\right]  \sqrt{2\pi\ }}{\sqrt{q\left\vert f^{(2)}(0)\right\vert }}\left[
1+\frac{f^{(4)}(0)}{72q\left\vert f^{(2)}(0)\right\vert }\right]  \label{L10}%
\end{equation}

By using Eqs.(\ref{L9}, \ref{L4}, \ref{L6}, \ref{L7}) we obtain%
\begin{align}
f^{(2)}(0)  &  =-\frac{\Gamma(\frac{3}{\alpha})}{\Gamma(\frac{1}{\alpha}%
)}\label{L11}\\
f^{(4)}(0)  &  =\frac{\Gamma(\frac{5}{\alpha})}{\Gamma(\frac{1}{\alpha}%
)}-3\left[  \frac{\Gamma(\frac{3}{\alpha})}{\Gamma(\frac{1}{\alpha})}\right]
^{2}\label{L12}\\
f(0)  &  =\log\left[  \frac{1}{\pi\alpha}\Gamma(\frac{1}{\alpha})\right]
\label{L13}%
\end{align}

From Eqs.(\ref{L10}-\ref{L13}) we obtain%
\begin{align*}
& S_{R,q}(\rho_{\alpha})\underset{q\rightarrow+\infty}{\asymp}\left[  \frac
{1}{\pi\alpha}\Gamma(\frac{1}{\alpha})\right]  ^{q}\left[  \frac{2\pi
\Gamma(1/\alpha)}{q\Gamma(3/\alpha)}\right]  ^{1/2}\times\\
& \left\{  1+\frac{1}{72q}\left[  \frac{\Gamma(5/\alpha)}{\Gamma(3/\alpha
)}-3\frac{\Gamma(3/\alpha)}{\Gamma(1/\alpha)}\right]  \right\}
\end{align*}

\bigskip

\section{ Conclusions}

\ In the perspective of the geometric interpretation of the R\'{e}nyi entropy
we obtained a duality relation. In general the duality principles are largely
used in the theory and applications of the convex optimization methods, so it
can be used in the interpretation and applications of the MaxEnt principles
for the \ R\'{e}nyi entropy, for the optimal approximation of \ the solution
of the Fokker-Planck equation. \newline

\ In the domain of the large values of the R\'{e}nyi parameter, when the
numerical methods does not work, we obtained an asymptotic expansion of the
R\'{e}nyi entropy for the symmetric stable L\'{e}vy distribution.

\bigskip
\end{document}